\documentclass[preprint,12pt]{elsarticle}

\usepackage[utf8]{inputenc}
\usepackage[T1]{fontenc}
\usepackage[margin=1in]{geometry}
\usepackage{amsmath}
\usepackage{amssymb}
\usepackage{amsthm}
\usepackage{mathtools}
\usepackage{graphicx}
\usepackage{float}
\usepackage{subcaption}
\usepackage{booktabs}
\usepackage{longtable}
\usepackage{array}
\usepackage{dsfont}
\usepackage{bbm}

\usepackage{hyperref}
\usepackage{xcolor}

\usepackage{graphicx}
\usepackage{amssymb}

\usepackage{lineno}

\theoremstyle{definition}
\newtheorem{definition}{Definition}
\newtheorem{proposition}{Proposition}
\newtheorem{remark}{Remark}

\hypersetup{
    colorlinks=true,
    linkcolor=blue!50!black,
    citecolor=blue!50!black,
    urlcolor=blue!50!black
}

\begin{document}

\begin{frontmatter}


\title{Network-Normative Belief Updating \\in High-Dimensional Ideological Space}

\author{Chico Q. Camargo \footnote{Department of Computer Science, University of Exeter, EX4 4QF, UK}}

\date{\today}


\begin{abstract}
\noindent Most mathematical models of opinion dynamics treat attitudes as scalar quantities or as positions on a low-dimensional ideological axis. Empirical attitudes, however, are bundles of positions across many policy issues, and the geometry of the resulting high-dimensional belief space is non-trivial. This paper develops a network-theoretic framework for analysing how individuals move through such a space between two waves of measurement. Continuous attitude profiles in $[0,1]^n$ are discretised onto regular grids of resolution $k$, occupied positions form a network whose adjacency is defined by single-issue unit moves, and densely populated regions of that network are interpreted as \emph{network-normative}: empirically common configurations of attitudes in the population.

We introduce a hierarchy of three null models against which observed movement can be benchmarked. A closed-form \emph{coverage baseline} $\pi_k = 1-(1-1/k^n)^N \approx 1-e^{-N/k^n}$ predicts the fraction of cells expected to be occupied at baseline under uniform random placement, with no behavioural parameters. A \emph{local random-walk} null retains each respondent's baseline position and asks whether destinations are over-represented in occupied regions relative to a uniform one- or two-step move. A \emph{marginal permutation} null preserves the per-issue distribution of opinion change while disrupting within-respondent cross-issue coupling. The three nulls form a hierarchy of increasing realism and isolate distinct contributions to apparent attraction.

Applying the framework to a two-wave panel of $N{=}1194$ respondents reporting attitudes on $n{=}10$ policy issues using continuous visual analogue scales, we obtain three findings. First, at the focal resolution $k{=}3$ the observed inside rate $\hat p_3 \approx 0.73$ exceeds the coverage baseline $\pi_3 \approx 0.02$ by a factor of $36$, exceeds the two-hop random-walk null by $\sim 0.30$, and exceeds the perturbation null by $\sim 0.04$; only the one-hop random walk is competitive (gap $\sim 0.06$). Second, the gap between the observed rate and each null evolves with $k$ in interpretable ways: the two-hop gap is large at every scale but decreases with $k$, the perturbation gap grows from near zero at $k{=}2$ to $\sim 0.14$ at $k{=}5$, and the one-hop gap is small throughout. Third, the closed-form coverage baseline accounts for the bulk of the slide $\hat p_k:\,0.93\!\to\!0.38$ as $k:\,2\!\to\!5$, but the empirical occupied fraction $|C_k|/k^n$ sits well below the analytic curve, indicating substantial baseline concentration. Network-normative attraction is therefore real but representation-contingent, and which null is exceeded changes systematically with $k$.
\end{abstract}

\end{frontmatter}

\section{Introduction}

Opinion dynamics has been a fertile ground for applied mathematics for several decades. Classical models treat individual attitudes as scalar quantities and study how they evolve under iterated social influence, including informational cascades~\cite{BikhchandaniSushilHirshleifer1992} and other forms of network-driven contagion~\cite{Watts2007, Lorenz2011}. This literature has generated rich theory about consensus, polarisation, and opinion clustering. It has also produced a robust qualitative empirical signature: when individuals are exposed to information about what other people think, aggregate attitudes tend to drift toward perceived majority positions, a so-called \emph{bandwagon} effect~\cite{mutz1992mass, nadeau1993new, marsh1985back, Kenney1996, Mehrabian1998}, while a non-trivial subset of the population responds in the opposite direction, a so-called \emph{underdog} or contrarian effect~\cite{Fleitas1971, Chatterjee2021}. Both phenomena coexist in real populations~\cite{Solomon1955, Bond1996}, and the aggregate behaviour reflects the mixture of conformist and reactant individuals.

Two simplifications, however, are present in nearly all of this literature. First, most studies of opinion change measure displacement on a single issue or along a one-dimensional left--right axis, even though real citizens hold structured bundles of views across multiple policy domains. A respondent can become more moderate on climate policy while becoming more polarised on immigration; a one-dimensional summary cannot recover this. Second, most studies report \emph{directional} change--the average shift toward or away from the cue--but say little about \emph{where} citizens move in attitude space. Two cohorts with the same mean shift can occupy very different geometric regions afterwards, with different implications for downstream collective behaviour.

This paper takes a different angle. Following work in computational social science that emphasises the role of network structure and high-dimensional state in shaping aggregate political behaviour~\cite{margetts2015political}, we represent each respondent as a point in $[0,1]^n$ where each coordinate encodes the level of agreement with one of $n$ policy statements. A two-wave panel measurement assigns to each respondent an origin and a destination in this space, and the question we ask is structural: when destinations are aggregated, do they concentrate in regions that were already densely populated at baseline? We call such regions \emph{network-normative} (in the sense of empirical commonness, not moral norm), and the structural-attraction hypothesis predicts that movement should disproportionately terminate inside them.

The hypothesis is interesting because at least three distinct mechanisms could produce it: social conformity to perceived majority configurations; cognitive pressure toward familiar bundles of issue positions; and interaction effects between issues that make some combinations more stable than others. It is also interesting because it is not trivially true. Two alternative processes can mimic apparent attraction without any structural pull. The first is \emph{local-neighbour diffusion}: if respondents take short steps from their baseline position, and nearby cells are unevenly populated, destinations may appear normative even without global pull. The second is \emph{regression to a central tendency}: if extreme baseline values soften over time, coarse discretisation can convert independent issue-level moderation into apparent movement toward dense regions. Distinguishing genuine attraction from these two artefacts is the central methodological problem the paper addresses.

The contribution is twofold. First, we develop a formal mathematical framework that places respondent profiles on a lattice graph, defines occupancy and normative regions, and articulates a hierarchy of null models against which empirical movement can be evaluated. The hierarchy ranges from a parameter-free coverage baseline that depends only on the geometry of the lattice and the sample size, through a local random-walk benchmark, to a permutation-based null that preserves marginal change distributions while disrupting within-respondent cross-issue coupling. Second, we apply the framework to a two-wave panel in which $N{=}1194$ respondents reported attitudes on $n{=}10$ policy issues using continuous visual analogue scales, with measurement at baseline and again roughly two days later.

The empirical analysis is organised around three findings, each of which the framework makes precise. \emph{(i)} At a focal resolution where the lattice has empirical structure but is not saturated by the sample, observed inside rates exceed the coverage baseline by a large factor, exceed the second-neighbour random walk by a sizeable margin, and exceed the marginal-permutation null by a small but positive amount; only the first-neighbour random walk is competitive. \emph{(ii)} The gap between the observed inside rate and each null evolves with the resolution $k$ in interpretable ways: the second-neighbour gap is large but decreasing in $k$, the perturbation gap grows from near-zero at coarse resolutions to non-trivial values at fine resolutions, and the first-neighbour gap is small throughout. \emph{(iii)} The closed-form coverage baseline accounts for the bulk of the resolution-driven slide of $\hat p_k$, while the systematic shortfall of the empirical occupied fraction $|C_k|/k^n$ relative to the analytic curve $\pi_k$ quantifies how much of the inside rate is carried by baseline concentration rather than lattice geometry.

Together, these findings replace the simpler claim ``attraction is real'' with a more refined statement: which null is exceeded, and by how much, changes systematically with $k$ in ways that the framework predicts. The methodological recommendation that follows is to compute and report each of the three nulls at every resolution, treating the gap profile rather than the inside rate as the empirically meaningful quantity.

The remainder of the paper is organised as follows. Section~\ref{sec:framework} sets up the mathematical framework. Section~\ref{sec:data} describes the data. Section~\ref{sec:results} presents the empirical results in three steps that mirror the three findings above. Section~\ref{sec:discussion} discusses the connection to coarse-graining ideas in statistical physics and to attractor landscapes in opinion-dynamics models. Section~\ref{sec:limitations} delimits the inferential scope of the analysis. Section~\ref{sec:conclusion} concludes.

\section{Mathematical Framework}\label{sec:framework}

\subsection{Belief state space and discretisation}

Let $n \in \mathbb{N}$ denote the number of policy issues under consideration. We represent each respondent's attitude profile as a point in the unit hypercube
\begin{equation}
    \mathcal{X} = [0,1]^n,
\end{equation}
where the $j$-th coordinate $a_j \in [0,1]$ encodes the respondent's level of agreement with the $j$-th statement, with $0$ indicating maximal disagreement and $1$ maximal agreement. Each coordinate is observed on a continuous visual analogue scale, so $\mathcal{X}$ is the natural domain for raw measurements.

Network-theoretic analysis requires a discrete state space. Fix a resolution $k \in \mathbb{N}$ with $k \geq 2$, partition $[0,1]$ into $k$ equal-width bins of width $1/k$, and represent each bin by its centroid. This yields the bin-centroid grid
\begin{equation}
    \mathcal{X}_k = \left\{\tfrac{1}{2k}, \tfrac{3}{2k}, \tfrac{5}{2k}, \ldots, \tfrac{2k-1}{2k}\right\}^n \subset \mathcal{X},
\end{equation}
with $k$ values per axis, step length $1/k$ between consecutive centroids, and total cell count $|\mathcal{X}_k| = k^n$. The discretisation map is the nearest-centroid projection
\begin{equation}
    \phi_k: \mathcal{X} \to \mathcal{X}_k, \qquad \phi_k(a) = \arg\min_{x \in \mathcal{X}_k} \|a-x\|_2,
\end{equation}
with ties broken by lexicographic order of coordinates. We will work with $k \in \{2,3,4,5\}$, corresponding to grid steps $\tfrac{1}{2}, \tfrac{1}{3}, \tfrac{1}{4}, \tfrac{1}{5}$ in the original measurement scale and reported in the empirical analysis as steps $0.50$, $0.33$, $0.25$, and $0.20$ respectively.

\subsection{Empirical occupation and the belief network}

Suppose the sample consists of $N$ respondents, each measured at two waves $t \in \{1,2\}$ (baseline and follow-up). Let $a_i^{(t)} \in \mathcal{X}$ denote the raw profile of respondent $i$ at wave $t$, and let $x_i^{(t)} = \phi_k(a_i^{(t)}) \in \mathcal{X}_k$ denote its discretised image at resolution $k$. The baseline empirical occupation measure is
\begin{equation}
    \mu_k(v) = \sum_{i=1}^{N}  \large\mathbbm{1}\!\left[x_i^{\text{origin}} = v\right], \qquad c \in \mathcal{X}_k,
\end{equation}
which counts the number of respondents whose discretised baseline profile coincides with cell $c$. The support of $\mu_k$,
\begin{equation}
    C_k = \mathrm{supp}(\mu_k) = \{c \in \mathcal{X}_k : \mu_k(c) > 0\},
\end{equation}
is the set of \emph{occupied cells} at baseline.

Adjacency in the discretised state space is defined by single-issue unit moves: two cells $u,v \in \mathcal{X}_k$ are adjacent if they differ in exactly one coordinate, and the difference equals one grid step $1/k$. Equivalently, define the rescaled Hamming distance $d_H(u,v) = \#\{j: u_j \neq v_j\}$ when both points lie on the same lattice; restricting to neighbours that differ by one grid step in a single coordinate gives the edge set
\begin{equation}
    E_k = \left\{(u,v) \in \mathcal{X}_k \times \mathcal{X}_k : \exists j \text{ such that } u_\ell = v_\ell \;\forall \ell \neq j \text{ and } |u_j - v_j| = \tfrac{1}{k}\right\}.
\end{equation}
The pair $\mathcal{G}_k = (\mathcal{X}_k, E_k)$ is the lattice graph of the discretised hypercube. The \emph{empirical belief network} $G_k$ is the induced subgraph on $C_k$. The resulting graph is illustrated in Figure~\ref{fig:discretisation} for the case of $N=3$ issues.

\begin{figure}[ht]
\centering
\includegraphics[width=0.98\linewidth]{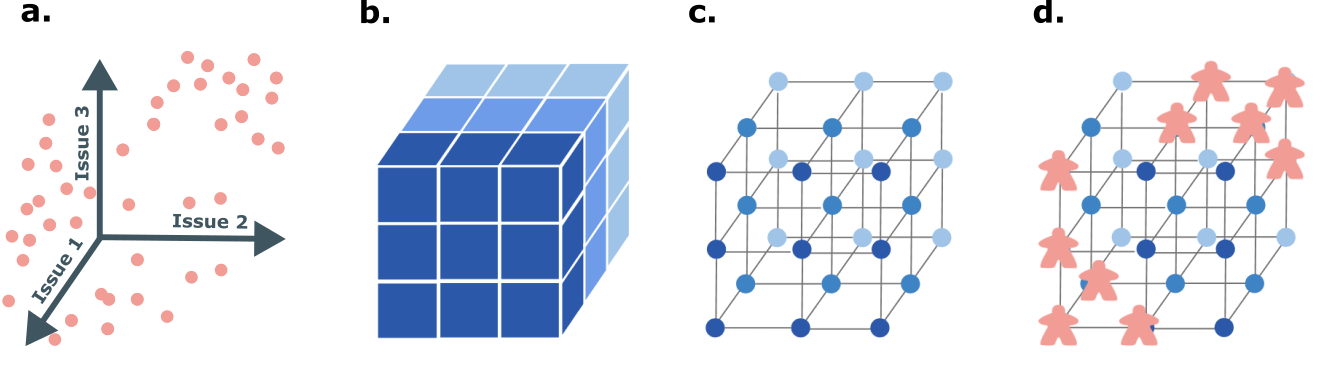}
\caption{\textbf{Construction of the belief-space graph.} (a) Each respondent is a point in $\mathcal{X} = [0,1]^n$; updating between waves is a transition in this space. (b) The bin-centroid lattice $\mathcal{X}_k$ at $k=3$. (c) the graph corresponding to the bin-centroid lattice. (d) The same graph, with empirical occupancy $\mu_k(v)$ visualised in three dimensions.}
\label{fig:discretisation}
\end{figure}

\begin{definition}[Network-normative regions]
For a threshold $\theta \in \mathbb{N}$, the network-normative region at resolution $k$ is
\begin{equation}
    \mathcal{N}_k(\theta) = \{c \in C_k : \mu_k(c) \geq \theta\}.
\end{equation}
The case $\theta=1$ recovers the full occupied set $C_k$ (\emph{untrimmed} occupancy); the case $\theta=2$ excludes singleton cells (\emph{trimmed} occupancy).
\end{definition}

The choice of $\theta$ is a modelling decision, not a feature of the data: trimmed occupancy tests whether structural conclusions depend on rare baseline configurations.

\subsection{Transitions and the inside rate}

For each respondent $i$, the discretised pair $(x_i^{\text{origin}}, x_i^{\text{destination}})$ is an origin--destination edge in $\mathcal{X}_k \times \mathcal{X}_k$, not necessarily a single-step edge in $E_k$. The object of interest here is the \emph{inside indicator}
\begin{equation}
    Y_i(k,\theta) = \large\mathbbm{1}\!\left[x_i^{\text{destination}} \in \mathcal{N}_k(\theta)\right],
\end{equation}
which records whether respondent $i$'s follow-up profile falls in a network-normative region defined at baseline. Aggregating over respondents gives the empirical inside rate
\begin{equation}
    \hat{p}_k(\theta) = \frac{1}{N} \sum_{i=1}^{N} Y_i(k,\theta).
\end{equation}

The question of network-normative attraction reduces to the comparison between $\hat{p}_k(\theta)$ and an appropriate null prediction. The remainder of this section develops three such null models.

\subsection{Null Model 1: analytic coverage baseline}\label{subsec:coverage}

The most parsimonious null model assumes that respondents at follow-up are placed independently and uniformly at random on the discretised grid $\mathcal{X}_k$, with no relationship to the baseline distribution. Under this null model, the probability that a randomly sampled cell in $\mathcal{X}_k$ is occupied at baseline is the coverage fraction.

\begin{proposition}[Coverage baseline]\label{prop:coverage}
Suppose $N$ respondents are placed at baseline independently and uniformly on $\mathcal{X}_k$. Then for any cell $v \in \mathcal{X}_k$,
\begin{equation}
    \mathbb{P}[v \in C_k] = 1 - \left(1 - \frac{1}{k^n}\right)^N \;=:\; \pi_k,
\end{equation}
and consequently the expected number of occupied cells is $E[|C_k|] = k^n \pi_k$. For $k^n \gg 1$,
\begin{equation}
    \pi_k \approx 1 - \exp\!\left(-\frac{N}{k^n}\right). \label{eq:coverage-approx}
\end{equation}
\end{proposition}

\begin{proof}
Consider a fixed cell $c \in \mathcal{X}_k$. We want $\mathbb{P}[c \in C_k]$, 
the probability that at least one of the $N$ respondents lands on $c$ after 
discretisation. It is easier to work with the complementary event. Under the 
null model, each respondent is placed independently and uniformly at random on 
$\mathcal{X}_k$, which has $k^n$ cells in total, so the probability that a 
single respondent misses $c$ is
\begin{equation}
    \mathbb{P}[\text{respondent } i \text{ misses } c] = 1 - \frac{1}{k^n}.
\end{equation}
Since respondents are placed independently, the probability that all $N$ of 
them miss cell $c$ is
\begin{equation}
    \mathbb{P}[c \notin C_k] = \left(1 - \frac{1}{k^n}\right)^N.
\end{equation}

The approximation~\eqref{eq:coverage-approx} is the standard limit $(1-x)^N \approx e^{-Nx}$ for small $x$.
\end{proof}

The coverage fraction $\pi_k$ is a closed-form, parameter-free function of $(k, n, N)$. It is also, under the uniform-baseline null model, the expected probability that a uniformly drawn follow-up cell falls inside the occupied set. Hence under this null model, the average empirical occupation rate for a normative network with $\theta = 1$ (i.e. at least one occupant per cell) is:
\begin{equation}
    E[\hat{p}_k(1)] = \pi_k,
\end{equation}
where the expectation is taken over both baseline. Equation~\eqref{eq:coverage-approx} predicts that the inside rate must collapse rapidly as $k$ grows, simply because $k^n$ grows much faster than $N$. For the empirical configuration $N=1194$, $n=10$, the ratio $N/k^n$ takes the values
\begin{equation}
    \frac{N}{2^{10}} \approx 1.17, \quad \frac{N}{3^{10}} \approx 2.02 \times 10^{-2}, \quad \frac{N}{4^{10}} \approx 1.14 \times 10^{-3}, \quad \frac{N}{5^{10}} \approx 1.22 \times 10^{-4},
\end{equation}
yielding analytic coverage predictions
\begin{equation}
    \pi_2 \approx 0.689, \quad \pi_3 \approx 0.0200, \quad \pi_4 \approx 0.00114, \quad \pi_5 \approx 0.000122. \label{eq:coverage-values}
\end{equation}
The coverage fraction $\pi_k$ sets the floor above which the inside rate becomes informative: any observed $\hat{p}_k$ at or below $\pi_k$ is 
consistent with purely random follow-up placement and carries no behavioural 
signal. At $k=2$, this floor sits at $\pi_2 \approx 0.69$, meaning an 
observed inside rate must clear $69\%$ before it tells us anything beyond 
lattice geometry. At $k=5$, the floor collapses to $\pi_5 \approx 10^{-4}$, 
so even a very small observed rate is already informative. The steep drop in 
$\pi_k$ as $k$ increases --- driven entirely by the growth of $k^n$ relative 
to $N$ --- therefore accounts mechanically for a large portion of the 
empirical scale gradient in $\hat{p}_k$: much of the slide from high inside 
rates at coarse resolution to low inside rates at fine resolution is a 
geometric inevitability, not a behavioural finding.
\begin{remark}
The coverage baseline assumes both baseline placement and follow-up placement are uniform on $\mathcal{X}_k$. If the baseline distribution is concentrated rather than uniform, the expected number of distinct occupied cells is smaller than $k^n \pi_k$ (because samples cluster), so a uniformly random follow-up draw lands inside the occupied set with probability \emph{below} $\pi_k$, not above. The drivers of empirical inside rates above $\pi_k$ are therefore (i) persistence of individual position between one time step and the next, which forces the inside indicator toward $1$ regardless of the size of $C_k$, and (ii) concentration of the follow-up distribution in the same region as the baseline, which is precisely the structural-attraction signal of interest. Disambiguating (i) from (ii) requires the more refined null models below, and is what the local random-walk and permutation comparisons are designed to do.
\end{remark}

\subsection{Null model 2: local random-walk baseline}\label{subsec:localrw}

The uniform coverage null model treats follow-up positions as independent of baseline positions. This is unrealistic: between two waves separated by a few days, individual positions are highly persistent. A more demanding null retains the baseline position of each respondent and assumes that movement, conditional on origin, is a random walk on the lattice graph.

For a cell $x \in \mathcal{X}_k$ and a radius $r \in \mathbb{N}_0$, define the $r$-hop reachable set
\begin{equation}
    \mathcal{R}_r(x) = \{y \in \mathcal{X}_k : d_{\mathcal{G}_k}(x,y) \leq r\},
\end{equation}
where $d_{\mathcal{G}_k}$ is the graph distance on the lattice $\mathcal{G}_k$. Equivalently, $\mathcal{R}_r(x)$ is the set of cells reachable from $x$ in at most $r$ single-issue unit moves.

\begin{proposition}[Reachable-set sizes]\label{prop:reach}
Suppose $k \geq 3$, and let $x \in \mathcal{X}_k$ be an \emph{interior} cell, defined as a cell whose bin index in every coordinate lies in $\{1, \ldots, k-2\}$ (so that one-step moves up and down are both available in every coordinate). Then
\begin{align}
    |\mathcal{R}_1(x)| &= 1 + 2n, \\
    |\mathcal{R}_2(x)| &= 1 + 2n + 2n(n-1) + 2n \;=\; 1 + 4n + 2n(n-1) \;=\; 1 + 2n(n+1), \label{eq:R2}
\end{align}
where the contributions in~\eqref{eq:R2} are: the origin itself (1); single-issue moves of one step (2n); pairs of distinct issues each moved by one step (2n(n-1), counting both signs); and single-issue moves of two steps in one issue (2n).
\end{proposition}

For boundary cells the sizes are smaller, since some moves are geometrically infeasible. For $n=10$ and $k \geq 3$, the upper bounds are $|\mathcal{R}_1| \leq 21$ and $|\mathcal{R}_2| \leq 221$.

The case $k=2$ is qualitatively different and must be handled separately. At $k=2$ each axis offers exactly one neighbour (the only other bin), so every cell has the same reachable-set sizes:
\begin{equation}
    |\mathcal{R}_1(x)| = 1 + n, \qquad |\mathcal{R}_2(x)| = 1 + n + \binom{n}{2}, \qquad k=2. \label{eq:R_k2}
\end{equation}
For $n=10$ this gives $|\mathcal{R}_1| = 11$ and $|\mathcal{R}_2| = 56$, rather than the interior values $21$ and $221$.

The local random-walk null assigns to respondent $i$ a destination drawn uniformly from $\mathcal{R}_r(x_i^{\text{origin}})$. Under this null, the probability of landing inside $\mathcal{N}_k(\theta)$ conditional on origin $x_i^{\text{origin}}$ is
\begin{equation}
    p_{r,k,\theta}^{\mathrm{loc}}\!\left(x_i^{\text{origin}}\right) = \frac{|\mathcal{R}_r(x_i^{\text{origin}}) \cap \mathcal{N}_k(\theta)|}{|\mathcal{R}_r(x_i^{\text{origin}})|}, \label{eq:p_loc}
\end{equation}
and the expected inside rate is
\begin{equation}
    E_{\mathrm{loc}}[\hat{p}_k(\theta)] = \frac{1}{N} \sum_{i=1}^N p_{r,k,\theta}^{\mathrm{loc}}\!\left(x_i^{\text{origin}}\right). \label{eq:E_loc}
\end{equation}
We compute equation~\eqref{eq:p_loc} empirically for each respondent at both $r=1$ (first-neighbour expectation) and $r=2$ (second-neighbour expectation), giving an individual-level benchmark that retains the baseline geometry.

\begin{remark}
The first-neighbour null model ($r=1$) is the most conservative: it asks whether a single-step move biases destinations toward $\mathcal{N}_k(\theta)$ above what the local lattice already imposes. The second-neighbour null model ($r=2$) is more permissive: it allows respondents to take up to two issue-coordinate moves and asks whether destinations are still over-represented in dense regions. Both are unbiased in the sense that, in expectation, a uniform random walker has inside rate equal to the average local coverage in equation~\eqref{eq:p_loc}.
\end{remark}

\subsection{Null model 3: marginal permutation}

The local random-walk null model is geometric: it does not use the empirical distribution of issue-level changes. To complement it, we introduce a permutation null that preserves the marginal distribution of opinion change on each issue while disrupting the within-respondent coupling of changes across issues.

For each issue $j \in \{1, \ldots, n\}$, let $\Delta_{ij} = a_{ij}^{\text{destination}} - a_{ij}^{\text{origin}}$ denote the observed change for respondent $i$ on issue $j$. Draw, independently for each issue, a uniformly random permutation $\sigma_j$ of the respondent indices, and set the perturbed follow-up profile
\begin{equation}
    \tilde{a}_{ij} = a_{ij}^{\text{origin}} + \Delta_{\sigma_j(i),\, j}.
\end{equation}
Each issue is permuted independently, so the marginal distribution of $\Delta$ on issue $j$ is preserved, while any joint structure across issues within respondents is destroyed. Discretising $\tilde{a}_i$ via $\phi_k$ and recomputing the inside rate yields a perturbed estimate $\tilde{p}_k(\theta)$. Repeating the procedure many times and comparing $\hat{p}_k(\theta)$ to the distribution of $\tilde{p}_k(\theta)$ provides a non-parametric test for whether within-respondent coupling of issue changes contributes to apparent attraction beyond what the marginal change distributions already imply. We use $200$ independent permutations and report the $2.5$ and $97.5$ percentiles of the resulting $\tilde{p}_k(\theta)$ distribution.

\subsection{Hierarchy of null models}

The three nulls form a hierarchy of increasing realism, summarised in Table~\ref{tab:nulls}. The coverage null model (1) ignores baseline persistence entirely and is parameter-free; the local random-walk null model (2) preserves baseline positions and lattice geometry but assumes uniform random destinations within a fixed radius; the permutation null model (3) preserves baseline positions and marginal change distributions but disrupts cross-issue coordination. A behavioural signal of network-normative attraction must, at minimum, exceed model 1; to be a genuine signal of \emph{coordinated} attraction it should exceed model 3.

\begin{table}[h]
\centering
\small
\begin{tabular}{p{2.8cm} p{5.4cm} p{7cm}}
\toprule
\textbf{Null} & \textbf{Preserves} & \textbf{Disrupts} \\
\midrule
1. Coverage & Lattice geometry; sample size $N$ & Baseline positions; persistence; cross-issue structure \\
2. Local RW & Baseline positions; lattice geometry; reachable-set structure & Direction of movement; cross-issue coupling \\
3. Permutation & Baseline positions; per-issue marginal change distributions & Within-respondent coupling of changes across issues \\
\bottomrule
\end{tabular}
\caption{The three null models compared to the observed inside rate $\hat{p}_k(\theta)$.}
\label{tab:nulls}
\end{table}

\section{Data}\label{sec:data}

The framework of Section~\ref{sec:framework} is applied to a two-wave panel originally designed to study the impact of opinion polls on individual policy preferences in the United Kingdom. The dataset is described in detail in \citep{pantazipublic}; we summarise here only the features relevant to the network analysis. Treatment and control conditions from the original design are pooled throughout: the present analysis is structural rather than causal, and any conditional differences between conditions are absorbed in the marginal change distribution that the perturbation null preserves.

\paragraph{Sample} Participants were recruited via the Prolific platform, with stratification matching the UK Office of National Statistics distribution on age, gender, and ethnicity. The full sample contained 1209 UK residents; 1174 completed the entire survey. After restricting to respondents observed at both waves with full responses on the ten focal issues, the analytical sample for the network analysis is $N=1194$. The mean age was $45.3$ years ($SD = 15.5$); $51.9\%$ identified as female, $47.9\%$ as male, and $0.2\%$ unspecified. Inclusion criteria required UK residency, UK citizenship, and age $\geq 18$.

\paragraph{Issues} Each respondent reported attitudes on $n=10$ policy issues spanning politically salient domains. Each issue was associated with a single statement; respondents indicated their level of agreement on a continuous visual analogue scale rendered as a vertical slider, with endpoints labelled ``Strongly disagree'' and ``Strongly agree''. The ten issues and their statements are listed in Table~\ref{tab:issues}.

\begin{table}[h]
\centering
\small
\begin{tabular}{p{3cm} p{12cm}}
\toprule
\textbf{Issue} & \textbf{Statement} \\
\midrule
Climate change & The government should be doing more to tackle climate change. \\
Social inequality & The government should take more steps to reduce social inequality. \\
Human rights & The government should do more to protect human rights regardless of race, religion etc. \\
Immigration & The government should reduce the level of Immigration into Britain. \\
LGBTQ+ rights & It is important to ensure LGBTQ+ people have the same rights as other members of society. \\
Women's rights & More should be done to ensure women have equal rights to men. \\
Misinformation & Social media websites should not display fake news stories. \\
Hate speech & I would support action to tackle online hate speech. \\
EU membership & Leaving the EU will be positive for the United Kingdom. \\
COVID-19 & During the COVID-19 crisis, the government should prioritise saving lives over protecting the economy. \\
\bottomrule
\end{tabular}
\caption{The ten policy issues and their corresponding statements presented to participants. Continuous responses on each issue are rescaled to $[0,1]$ to form the coordinates of $a_i^{(t)} \in \mathcal{X}$.}
\label{tab:issues}
\end{table}

\paragraph{Experimental design} The experiment used a two-wave panel design with within-subject repeated measurement, the two waves separated by approximately two days. At wave 1, all participants reported their attitudes on the ten issues; at wave 2, all participants again reported their attitudes on the same ten issues. The original design also included a between-subjects manipulation in which a randomly assigned subgroup received opinion-poll information between waves; that manipulation is not part of the present analysis.

\paragraph{Network construction inputs} For the analyses in this paper, the relevant inputs from the experiment are the pairs of attitude profiles $(a_i^{\text{origin}}, a_i^{\text{destination}}) \in [0,1]^{10} \times [0,1]^{10}$. Auxiliary measurements (issue-importance rankings, the Ten-Item Personality Inventory, perceived influence items) are not used. The analytical pipeline applies $\phi_k$ for $k \in \{2,3,4,5\}$, constructs $\mu_k$ and $\mathcal{N}_k(\theta)$ for $\theta \in \{1,2\}$, computes the inside indicator $Y_i(k,\theta)$, and benchmarks $\hat{p}_k(\theta)$ against the three nulls of Section~\ref{sec:framework}.

\paragraph{Properties relevant to the framework} Two features of the data make it suitable for the proposed analysis. First, the continuous visual analogue scales preserve fine-grained variation, allowing the discretisation map $\phi_k$ to be applied at multiple resolutions without the user-imposed coarseness of a five- or seven-point Likert scale. Second, the two-wave structure yields explicit origin--destination pairs in $\mathcal{X}$, which are the basic unit of the transition analysis.

\paragraph{Marginal structure} The original analysis of these data found that issue-level opinion change does not follow a Gaussian distribution but rather a Laplace distribution, with sharper peak and heavier tails (Figure~\ref{fig:laplace} in the appendix; pooled MLE scale $b \approx 0.087$). This matters for the network analysis in two ways. First, it implies that large displacements are not vanishingly rare: respondents do occasionally move several grid steps in a single issue, so transitions are not confined to immediate neighbourhoods of $x_i^{\text{origin}}$. Second, it justifies retaining the full distribution of changes in the permutation null model, rather than truncating tails on the Gaussian assumption that would otherwise discard a non-negligible fraction of observations.

\section{Results}\label{sec:results}

This section presents the empirical analysis in three steps that mirror the three findings stated in the introduction. Subsection~\ref{subsec:results-focal} establishes the lead claim at the focal resolution $k=3$. Subsection~\ref{subsec:results-scale} examines how the gap between the observed inside rate and each null evolves with $k$. Subsection~\ref{subsec:results-geometry} grounds the scale dependence in the analytic coverage baseline. Throughout we report results under untrimmed occupancy ($\theta=1$) in the main text; trimmed-occupancy results are summarised in~\ref{app:robustness} and noted where they differ qualitatively.

\subsection{Network normativity at the focal scale}\label{subsec:results-focal}

We choose $k=3$ as the focal resolution: the bin-centroid lattice has $3^{10} = 59{,}049$ cells and the sample of $N=1194$ occupies $|V_3| = 359$ of them, giving an empirical occupied fraction $|V_3|/3^{10} \approx 0.61\%$ that is small enough for the inside indicator to be informative but large enough for the framework to have empirical structure to read off. Figure~\ref{fig:focal} reports the four nulls at this resolution.

\begin{figure}[H]
\centering
\includegraphics[width=1.0\linewidth]{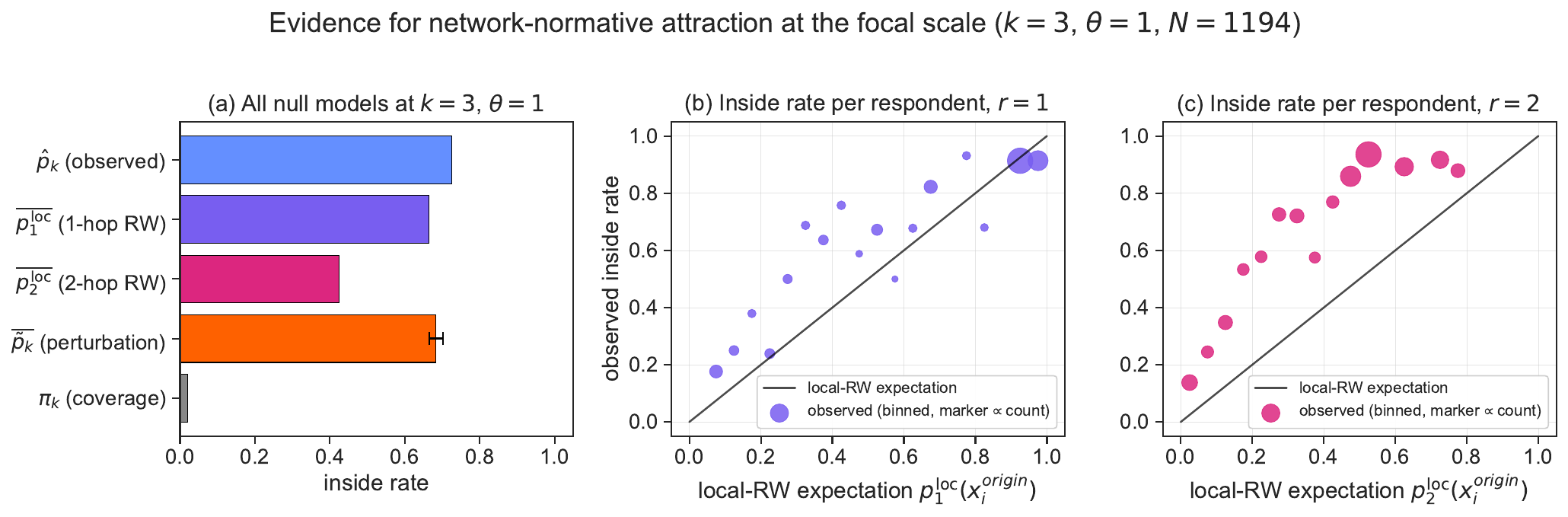}
\caption{\textbf{Network-normative attraction at the focal scale ($k=3$, $\theta=1$, $N=1194$).} (a) Observed inside rate $\hat p_3$ alongside the two local random-walk expectations $\overline{p^{\rm loc}_1}$ and $\overline{p^{\rm loc}_2}$, the perturbation null mean $\overline{\tilde p_3}$ (with $95\%$ band from $200$ permutations), and the analytic coverage baseline $\pi_3$. (b)--(c) Per-respondent observed inside rates as a function of the local random-walk expectation $p^{\rm loc}_r(x_i^{\text{origin}})$ for $r=1$ and $r=2$, binned along the $x$-axis with marker size proportional to bin count.}
\label{fig:focal}
\end{figure}

The numerical readings from Figure~\ref{fig:focal}(a) are:
\begin{align*}
    \hat p_3 &= 0.725, & \overline{p^{\rm loc}_1} &= 0.665, & \overline{p^{\rm loc}_2} &= 0.424, \\
    \overline{\tilde p_3} &= 0.684 \;[0.665, 0.702], & \pi_3 &= 0.020.
\end{align*}
The observed inside rate $\hat{p}_3$ exceeds the coverage baseline $\pi_3$ by a factor of $36$, exceeds the second-neighbour random-walk expectation $\overline{p^{\rm loc}_2}$ by $0.30$, and exceeds the first-neighbour random walk $\overline{p^{\rm loc}_1}$ by $0.06$, the perturbation null model $\overline{\tilde p_3}$ by $0.04$ (well outside the $95\%$ band).

The per-respondent panels of Figure~\ref{fig:focal} show that this aggregate gap is not driven by a few outliers. Almost every binned aggregate sits at or above the diagonal, with the strongest deviation at intermediate $p^{\rm loc}_r$ where the local-random null is most informative. For $r=2$ the binned points are well above the diagonal across the full $[0,1]$ range; for $r=1$ they lie close to but slightly above the diagonal, consistent with the small aggregate gap above $\overline{p^{\rm loc}_1}$. Trimmed occupancy at the same resolution gives larger gaps in every comparison ($\hat p_3 = 0.651$ vs.\ $\overline{p^{\rm loc}_1} = 0.533$, $\overline{p^{\rm loc}_2} = 0.273$, $\overline{\tilde p_3} = 0.561$), so the conclusion at this focal scale is conservative under untrimmed occupancy.
\\
\\
The empirical reading is therefore:
\begin{itemize}
    \item Network-normative attraction is real at the focal scale ($k=3$):  respondents land inside $\mathcal{N}_3(\theta)$ much more often than the lattice-geometry coverage baseline predicts and much more often than a uniform two-step random walk from baseline would imply.
    \item It is not merely persistence: a permutation that breaks within-respondent cross-issue coupling reduces the inside rate, and the gap is statistically clean.
    \item The first-neighbour random walk is, however, hard to clear: the observed inside rate exceeds it by only $\sim 0.06$. Local diffusion accounts for a substantial fraction of inside classification at this resolution, and stronger claims about coordinated long-range attraction require either finer resolutions (where the gaps with all nulls evolve in the directions discussed below) or alternative neighbourhood baselines.
\end{itemize}

\subsection{How the gaps depend on the resolution}\label{subsec:results-scale}

The behaviour of $\hat p_k$ across $k$ is interesting only relative to a baseline. Figure~\ref{fig:scale} reports all four statistics across resolutions, alongside the gap profile that summarises which nulls are exceeded at each scale.

\begin{figure}[H]
\centering
\includegraphics[width=0.98\linewidth]{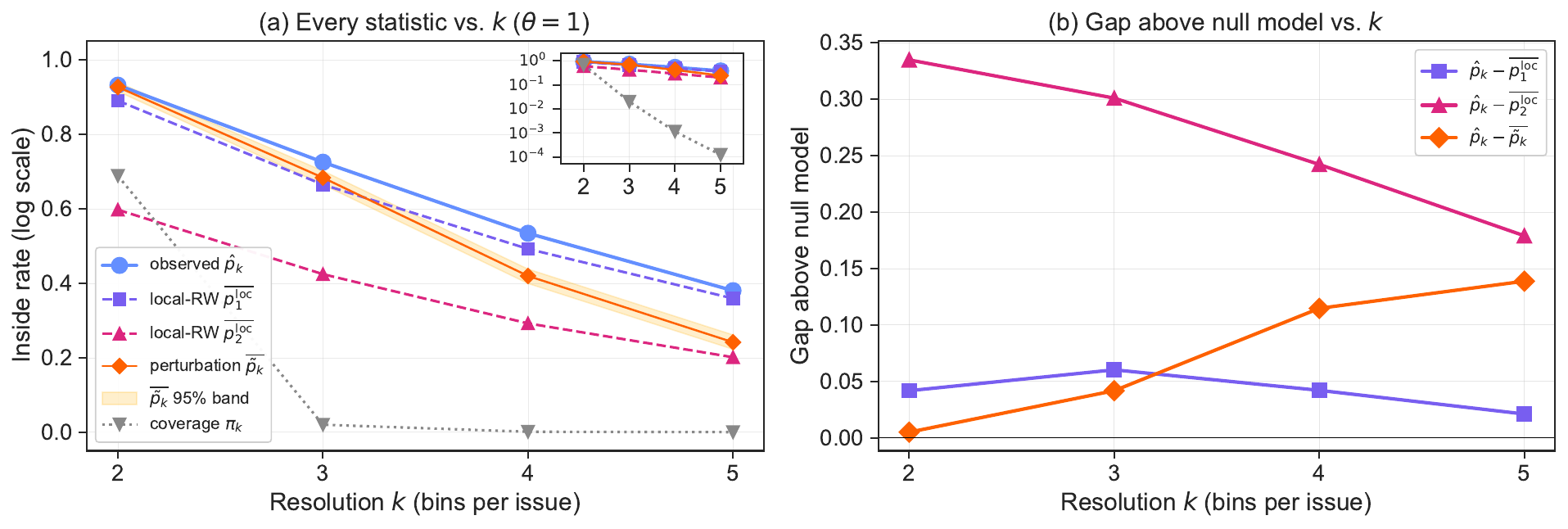}
\caption{\textbf{Scale dependence ($N=1194$, $n=10$, $\theta=1$).} (a) Observed inside rate $\hat p_k$, local random-walk expectations $\overline{p^{\rm loc}_1}, \overline{p^{\rm loc}_2}$, perturbation null model mean $\overline{\tilde p_k}$ with $95\%$ band, and analytic coverage baseline $\pi_k$, all on a logarithmic $y$-axis. The inset shows the same plot, but in log scale. (b) Gap above null model: $\hat p_k - \overline{p^{\rm loc}_1}$, $\hat p_k - \overline{p^{\rm loc}_2}$, and $\hat p_k - \overline{\tilde p_k}$ as functions of $k$.}
\label{fig:scale}
\end{figure}

\noindent Three patterns emerge.

\paragraph{Coverage gap} The observed inside rate sits orders of magnitude above the analytic coverage baseline $\pi_k$ at all $k$, with the ratio $\hat p_k / \pi_k$ growing from $1.35$ at $k=2$ to roughly $3120$ at $k=5$. The coverage null is therefore always exceeded--but most decisively at fine resolution, where $\pi_k$ is essentially zero and any non-trivial inside rate is informative.

\paragraph{Random-walk gaps}
The first-neighbour gap $\hat p_k - \overline{p^{\rm loc}_1}$ is small at every resolution ($0.04$ at $k=2$, $0.06$ at $k=3$, $0.04$ at $k=4$, $0.02$ at $k=5$), meaning that a respondent taking a single random step from their baseline position would land inside the normative region almost as often as respondents actually do. Local diffusion alone therefore accounts for most of the observed inside rate at every scale. The second-neighbour gap $\hat p_k - \overline{p^{\rm loc}_2}$ tells a different story: it is large ($0.33$ at $k=2$) and decreases steadily with $k$ ($0.30$, $0.24$, $0.18$), reflecting the fact that allowing two steps exposes far more of the lattice to the random walker, and the difference between random two-step movement and actual concentrated movement is largest when the lattice is densely populated.

\paragraph{Perturbation gap}
The perturbation gap $\hat p_k - \overline{\tilde p_k}$ is the most diagnostic of the three. It is negligible at $k=2$ ($0.005$) and grows steadily to $0.04$ at $k=3$, $0.11$ at $k=4$, and $0.14$ at $k=5$. The reason is geometric: at coarse resolution, cells are wide relative to a typical wave-to-wave displacement, so whether a respondent's issue changes are coupled or independent makes little difference to which cell they end up in. As the lattice refines, cells narrow and cross-issue coordination starts to determine cell membership. The perturbation null model, which breaks within-respondent coupling while preserving marginal change distributions, therefore finds nothing to detect at $k=2$ but an increasingly clear signal at $k=4$ and $k=5$. This is a substantive finding in its own right: coupled cross-issue updating is only detectable at fine resolutions, precisely where  observed inside rates are well below $1$ but still well above $\pi_k$.

Trimmed-occupancy results (Appendix~\ref{app:robustness}, Figure~\ref{fig:appendix-A1}) preserve all three patterns and amplify the magnitudes: the second-neighbour gap reaches $0.54$ at $k=2$ trimmed, and the perturbation gap reaches $0.16$ at $k=4$ trimmed. The qualitative scale dependence is therefore not a function of the occupancy rule.

\subsection{Why the gradient looks the way it does}\label{subsec:results-geometry}

The previous two subsections established that the observed inside rate exceeds each null model by margins that depend systematically on $k$. But they left open a more basic question: why does $\hat{p}_k$ itself slide so steeply across resolutions, dropping from $0.93$ at $k=2$ to $0.38$ at $k=5$? One might attribute this entirely to behaviour --- respondents becoming less likely to land in normative regions as the lattice refines --- but the coverage baseline $\pi_k$ suggests a simpler explanation. Most of this variation in $\hat{p}_k$ is a geometric inevitability: as $k$ grows, the number of cells $k^n$ grows much faster than the sample size $N$, so the lattice becomes increasingly sparse and the normative region $\mathcal{N}_k(\theta)$ shrinks as a fraction of the total grid. This subsection makes that argument precise, using the analytic coverage baseline to decompose the empirical gradient into a geometric component and a residual that reflects genuine concentration of attitudes in the population.

Figure~\ref{fig:geometry} makes this decomposition visible. Panel (a) plots  the analytic coverage baseline $\pi_k$ as a continuous function of $k$ for  several values of $N$, with the experimental sample size $N=1194$ highlighted. The curve drops sharply as $k$ increases past the crossover $k^* = N^{1/n} \approx 2.03$, marking the transition from the saturation regime, where $k^n \lesssim N$ and almost every cell is visited, to the sparse regime, where $k^n \gg N$ and most cells are empty. Overlaid on the same panel are the empirical occupied fractions $|V_k|/k^n$ at $k \in \{2,3,4,5\}$, which sit consistently \emph{below} the analytic curve at every resolution.
This shortfall reflects the fact that real attitudes are not uniformly distributed: respondents cluster in certain regions of $\mathcal{X}_k$, so the $N$ samples land repeatedly on the same cells and visit fewer distinct cells than uniform random placement would predict. 
Panel (b) places the experiment in the broader $(k, N)$ plane, with contour lines at $\pi_k = 0.01$ and $\pi_k = 0.99$ demarcating three regimes: saturated (where the inside rate is near one regardless of behaviour), vanishing (where any non-zero inside rate is informative), and an informative band in between. The experimental track passes through the informative band at $k \in \{3,4,5\}$ and enters the saturated regime at $k=2$, confirming that the focal resolutions are  appropriately chosen for detecting a behavioural signal.

\begin{figure}[H]
\centering
\includegraphics[width=1.0\linewidth]{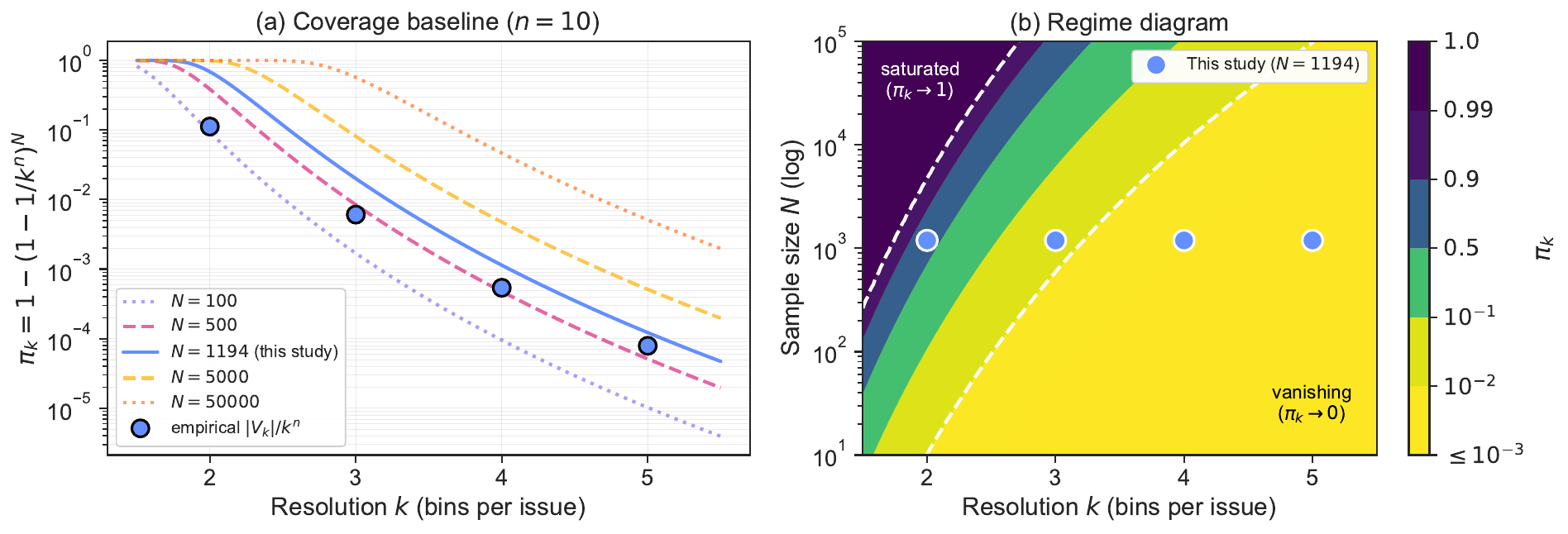}
\caption{\textbf{Geometric account of scale dependence ($n=10$).} (a) The analytic coverage baseline $\pi_k = 1 - (1-1/k^n)^N$ as a function of $k$ (treated continuously) for several values of $N$, with the experimental sample size highlighted; the empirical occupied fractions $|C_k|/k^n$ at $k\in\{2,3,4,5\}$ are overlaid as filled circles. (b) The same coverage baseline as a heat map over $(k, N)$, with the experimental track and contour lines at $\pi_k = 0.01$ and $\pi_k = 0.99$.}
\label{fig:geometry}
\end{figure}

On Figure~\ref{fig:geometry}(a), the coverage baseline $\pi_k(N=1194,n=10)$ slides from $0.69$ at $k=2$ to $1.2 \times 10^{-4}$ at $k=5$. The crossover $k^* = N^{1/n} \approx 2.03$ sits between $k=2$ and $k=3$, marking the transition between the saturation regime ($N/k^n \gtrsim 1$, $\pi_k \approx 1$) and the sparse regime ($N/k^n \ll 1$, $\pi_k \approx 0$). The qualitative shape of the empirical inside rate $\hat p_k$ tracks this geometric prediction, showing that empirical occupancy is below uniform. 
Still on Figure~\ref{fig:geometry}(a), the empirical points $|C_k|/k^n$ sit \emph{below} the analytic curve at every resolution by a non-trivial factor: $115/1024 = 11.2\%$ vs.\ $\pi_2 \approx 68.9\%$ (factor $6.1\times$); $359/59{,}049 = 0.61\%$ vs.\ $\pi_3 \approx 2.0\%$ (factor $3.3\times$); $565/4^{10} = 0.054\%$ vs.\ $\pi_4 \approx 0.114\%$ (factor $2.1\times$); $771/5^{10} = 0.0079\%$ vs.\ $\pi_5 \approx 0.012\%$ (factor $1.5\times$). The shortfall arises because real attitudes are not uniformly distributed on $\mathcal{X}_k$: they cluster, so the $N$ samples land repeatedly on the same cells and visit fewer distinct cells than uniform random placement would. The concentration ratio $\pi_k k^n / |C_k|$ decreases with $k$ as the lattice becomes too sparse for clustering to matter.

The regime diagram on Figure~\ref{fig:geometry}(b) places the experiment in the broader $(k, N)$ plane. The two contour lines $\pi_k = 0.01$ and $\pi_k = 0.99$ separate three regimes: a saturated regime where $\pi_k \approx 1$ and the inside rate carries little information about behaviour, a vanishing regime where $\pi_k \approx 0$ and any non-zero inside rate is informative, and an informative band in between. For $n=10$, the $N=1194$ track passes through the informative band at $k\in\{3,4,5\}$ and into the saturated regime at $k=2$.

\subsection{Summary}\label{subsec:results-summary}

The three findings together support a layered conclusion. Network-normative attraction is real: $\hat p_k$ exceeds the coverage baseline by a factor that ranges from $1.4$ at $k=2$ to over $3{,}000$ at $k=5$, exceeds the second-neighbour random walk by $0.18$--$0.33$, and exceeds the perturbation null by margins that grow with $k$. It is not merely a property of where one slices the lattice: each null is exceeded, and the gap profile across $k$ is stable in its qualitative shape. But the strength of the evidence depends on which null one cares about. The first-neighbour random walk is hard to clear at any $k$, suggesting that local diffusion is a substantial part of the story; the perturbation null is hardest to clear at coarse resolution and easiest at fine resolution, indicating that coupled cross-issue updating is a fine-resolution phenomenon. Treating the framework's output as a single inside rate at a single $k$ obscures all of this; the gap profile across $k$ and different null models is the natural observable.

\section{Discussion}\label{sec:discussion}

\subsection{Scale dependence}

A primary consequence of the framework in Section~\ref{sec:framework} is that scale dependence of network-normative attraction is, in part, a mathematical inevitability. The coverage baseline $\pi_k = 1 - (1 - 1/k^n)^N \approx 1 - e^{-N/k^n}$ depends only on the geometry of the lattice and the sample size; it predicts a sharp drop in expected inside rates as $k$ grows past the regime $k^n \sim N$. The empirical scale gradient observed in Section~\ref{sec:results} mirrors this prediction qualitatively, because the experiment's parameters $(N, n) = (1194, 10)$ place the threshold $k^* \approx N^{1/n} \approx 2.03$ between $k=2$ and $k=3$. Any analysis of high-dimensional belief space that uses fixed-resolution discretisation will encounter the same crossover, and any reported ``attraction to common configurations'' must be interpreted relative to where on the curve $\pi_k(N)$ the chosen resolution sits.

This observation is not specific to attitudes. The coverage formula is structurally identical to the classical occupancy problem in combinatorial probability, and the same expression appears in different guises across urn models, species-accumulation problems, and coupon-collector arguments. Translating it into the language of opinion dynamics gives a clean prescription: when partitioning a high-dimensional attitude space into $k^n$ cells, the analyst should report $N/k^n$ alongside any inside-rate statistic, since this ratio determines the coverage baseline and hence the level above which a substantive interpretation can begin.

\subsection{A renormalisation-group perspective}

The scale-dependent behaviour identified here is reminiscent of coarse-graining in statistical physics. In a renormalisation-group picture, observables computed at one resolution are not simply rescaled versions of those computed at a finer resolution: relevant and irrelevant degrees of freedom emerge as one moves up or down the scale ladder, and apparent macroscopic regularities (basins of attraction, ordered phases) can be artefacts of the resolution at which the system is observed. 
The belief network exhibits similar logic. The underlying system is fixed: the same $N$ respondents, the same raw attitude profiles, the same wave-to-wave changes. What varies across $k$ is only the resolution at which that system is discretised. Increasing $k$ is analogous to running the renormalisation group in reverse --- moving from coarse to fine --- and watching how the observable $\hat{p}_k$ changes as a result. The coverage baseline $\pi_k$ plays the role of a trivial reference point: it describes what a structureless system would look like at each scale, given only the geometry of the lattice and the sample size. The gap $\hat{p}_k - \pi_k$ is then the deviation from this reference, and its non-trivial shape across $k$ is the analogue of a non-trivial flow away from the trivial fixed point.

What this suggests is that network-normative attraction is best understood as a multi-scale property. Single-resolution reports are inherently ambiguous; a complete characterisation requires the curve $\hat{p}_k(\theta)$ as $k$ varies, together with the analytic curve $\pi_k$ for the same $(n, N)$. The gap between these two curves, rather than the value of either at a single $k$, is the empirically meaningful quantity.

\subsection{The perturbation gap and cross-issue coupling}

The most surprising empirical finding from the framework is the resolution-dependent behaviour of the perturbation gap $\hat p_k - \overline{\tilde p_k}$, which is essentially zero at coarse resolution and grows steadily with $k$. Because the perturbation null preserves marginal change distributions while disrupting within-respondent cross-issue coupling, this pattern says that coupled cross-issue updating--respondents adjusting several positions in a coordinated way that jointly tracks normative regions--is invisible at $k=2$ but accounts for $\sim 0.14$ of the inside rate at $k=5$.

The mechanical reason is that small cross-issue coordination cannot move respondents across cell boundaries when cells are large. At $k=2$ each cell occupies a hyper-octant of the unit cube, and a typical wave-to-wave displacement (with the Laplace distribution scale $b \approx 0.087$ on each issue) is small relative to the cell width $1/k = 1/2$. At $k=5$ the cell width is $0.20$ and the same displacements regularly cross cell boundaries, so the geometry of the displacement vector--whether issue changes are coupled or independent--starts to matter for cell membership. The framework therefore detects coupling exactly where coupling has room to act geometrically.

This reverses a naive reading of the results. Coarse resolution gives the largest \emph{absolute} inside rate but the least \emph{informative} signal, because most of the rate is reproducible from marginal change distributions alone. Fine resolution gives a smaller absolute rate but a cleaner signal of coordinated movement. The methodological consequence is that the perturbation null is not just a sanity check; it is the diagnostic that distinguishes scale regimes in which inside classification is dominated by geometry from regimes in which it carries information about behavioural coupling.
\section{Limitations}\label{sec:limitations}

Several limitations delimit the inferential scope of the analysis. First, discretisation is unavoidable for finite-state network modelling, and different resolutions induce different lattice topologies. We have made this dependence the central object of analysis rather than an obstacle; nonetheless, conclusions about absolute inside rates remain conditional on a chosen $(k,\theta)$ pair. Second, two-wave data identify origin and destination but not the intermediate trajectory. Higher-frequency longitudinal data would permit the explicit reconstruction of paths in $\mathcal{X}_k$ and a direct test of the random-walk null at the trajectory level rather than at the destination level. Third, the local random-walk null and the marginal-permutation null are informative counterfactuals but not full behavioural models: they cannot, by themselves, separate psychological conformity from strategic adjustment from measurement artefact. Fourth, the analysis is carried out on attitudinal outcomes; results should not be directly generalised to behavioural outcomes such as voter turnout, vote choice, or collective action. Fifth, the baseline distribution of attitudes in the sample is markedly non-uniform--most respondents support the progressive framings of the statements--so the network-normative regions $\mathcal{N}_k(\theta)$ are concentrated in a limited part of $\mathcal{X}_k$. Sample compositions with different ideological balance would induce different normative regions and potentially different scale-dependence curves. Sixth, and most importantly, all empirical findings reported here derive from a single dataset collected in the United Kingdom at a specific moment in time, on a specific set of ten policy issues. Whether the scale-dependence curves, gap profiles, and perturbation patterns documented here generalise to other populations, other issue sets, or other political contexts remains an open question. The framework itself is general, but its empirical conclusions are not: replication across datasets with different sample compositions, different numbers of issues, and different cultural and political settings is a necessary next step before any of the quantitative findings can be treated as robust.

\section{Conclusion}\label{sec:conclusion}

This paper has developed a network-theoretic framework for analysing belief updating in high-dimensional attitude space, applied it to a two-wave panel of $N=1194$ respondents on $n=10$ policy issues, and used the framework to address a methodological problem implicit in any network analysis of opinion change: how much of the apparent attraction to densely populated regions is a behavioural signal, and how much is a combinatorial property of the lattice on which the analysis is performed?

The framework introduces three null models that together separate these effects. The analytic coverage baseline $\pi_k = 1 - (1 - 1/k^n)^N$ is closed-form, parameter-free, and explains the bulk of the empirical scale gradient observed across resolutions. The local random-walk benchmark, computed at first- and second-neighbour radii, retains baseline persistence and lattice geometry, and identifies attraction beyond what local diffusion would explain. The marginal-permutation null preserves issue-level change distributions while disrupting cross-issue coupling, and isolates the contribution of coordinated within-respondent movement. The hierarchy provides both a benchmarking procedure for empirical work and a typology for opinion-dynamics models.

Empirically, the analysis shows that evidence for network-normative attraction is real but its strength depends on which null one cares about. The observed inside rate exceeds the coverage baseline by orders of magnitude at fine resolution, exceeds the second-neighbour random walk by sizeable margins at every resolution, and exceeds the perturbation null at fine resolution but only marginally at coarse resolution. The first-neighbour random walk is competitive with observed movement throughout. The substantive conclusion is that network-normative attraction is best characterised as a multi-statistic, multi-scale property, with the gap profile $\bigl(\hat p_k - \text{null model}\bigr)$ rather than the inside rate itself as the natural observable. These findings are, however, grounded in a single dataset, and their generalisability is not yet established. Validating the gap profiles and scale-dependence curves across datasets with different sample compositions, issue sets, and political contexts is a priority for future work.

Future work should extend the framework along four directions. First, higher-frequency longitudinal data would permit reconstruction of intermediate trajectories in $\mathcal{X}_k$ and direct testing of the random-walk null at the path level. Second, the framework can be paired with generative opinion-dynamics models simulated under matched parameters, allowing model-data comparisons that respect the lattice geometry. Third, sample-composition effects on the normative regions $\mathcal{N}_k(\theta)$ can be studied by analysing how scale-dependence curves shift under reweighting toward different ideological balances. Fourth, and most pressingly, the framework should be applied to additional two-wave panels drawn from different national contexts, issue domains, and time periods, to establish whether the quantitative patterns reported here --- the magnitude of the perturbation gap, the competitive performance of the first-neighbour null, the crossover at $k^* \approx N^{1/n}$ --- are stable empirical regularities or features specific to this sample. Each of these extensions builds on the same mathematical core: the lattice-graph representation of high-dimensional attitudes, the empirical occupancy measure as a stand-in for an attractor landscape, and the null-model hierarchy as the lens through which both empirical patterns and theoretical models should be evaluated.

\section{Methods}\label{sec:methods}

Attitudes were measured on continuous visual analogue scales over ten policy issues at baseline and follow-up, yielding raw profiles $a_i^{(t)} \in [0,1]^{10}$ for $i \in \{1, \ldots, 1194\}$ and $t \in \{1,2\}$. Profiles were discretised at four resolutions $k \in \{2,3,4,5\}$ via the nearest-bin-centroid projection $\phi_k$ defined in Section~\ref{sec:framework}; ties were broken by lexicographic order of coordinates. Baseline empirical occupation $\mu_k(v) = \sum_i \mathbbm{1}[\phi_k(a_i^{\text{origin}}) = v]$ was computed on the resulting grids, and the occupied set $C_k = \mathrm{supp}(\mu_k)$ was retained as the vertex set of the empirical belief network. Adjacency was defined by single-issue unit moves on the lattice. Network-normative regions $\mathcal{N}_k(\theta)$ were computed under both untrimmed ($\theta=1$) and trimmed ($\theta=2$) occupancy rules. Follow-up profiles $\phi_k(a_i^{\text{destination}})$ were classified as inside $\mathcal{N}_k(\theta)$ or outside.

The coverage baseline $\pi_k$ was computed analytically from Proposition~\ref{prop:coverage} for each $k$ at the empirical sample size $N=1194$. The local random-walk null was computed by enumerating, for each respondent, the reachable set $\mathcal{R}_r(\phi_k(a_i^{\text{origin}}))$ for $r \in \{1, 2\}$ on the discretised lattice, intersecting with $\mathcal{N}_k(\theta)$, and computing the per-respondent expected inside probability $p_{r,k,\theta}^{\mathrm{loc}}$ from equation~\eqref{eq:p_loc}. Boundary cells were handled by excluding moves that would leave the unit hypercube. The expected inside rate $E_{\mathrm{loc}}[\hat{p}_k(\theta)]$ was obtained by averaging $p_{r,k,\theta}^{\mathrm{loc}}$ over respondents.

The perturbation null was computed by drawing, for each issue $j$, an independent uniform random permutation $\sigma_j$ of the respondent indices, and constructing perturbed follow-up profiles $\tilde{a}_{ij} = a_{ij}^{\text{origin}} + (a_{\sigma_j(i),j}^{\text{destination}} - a_{\sigma_j(i),j}^{\text{origin}})$. The full pipeline (discretisation, occupancy, inside/outside classification) was rerun on the perturbed data over $200$ independent permutations; we report the perturbation mean and the $2.5$/$97.5$ percentiles.

\section*{Acknowledgements}
The author thanks the Alan Turing Institute (EPSRC grant EP/N510129/1) and the Volkswagen Foundation.


\newpage

\appendix
\section{Supplementary Robustness Figures}\label{app:robustness}

This appendix collects the full set of robustness visualisations underlying the main-text claims. Figure~\ref{fig:appendix-A1} shows the inside-rate statistics across all $(k, \theta)$ combinations; Figure~\ref{fig:appendix-A2} shows per-respondent observed-vs-expected scatter at every $(k, \theta, r)$ combination; Figure~\ref{fig:laplace} shows the pooled distribution of issue-level opinion change with a Laplace MLE fit, justifying the heavy-tailed distributional assumption that motivates the perturbation null; Figure~\ref{fig:marginals} shows the per-issue attitude distributions at each wave; Figure~\ref{fig:lattice-stats} shows lattice-statistics summaries; and Figure~\ref{fig:perturbation-dists} shows the histograms of perturbation-null inside rates.

\begin{figure}[H]
\centering
\includegraphics[width=0.98\linewidth]{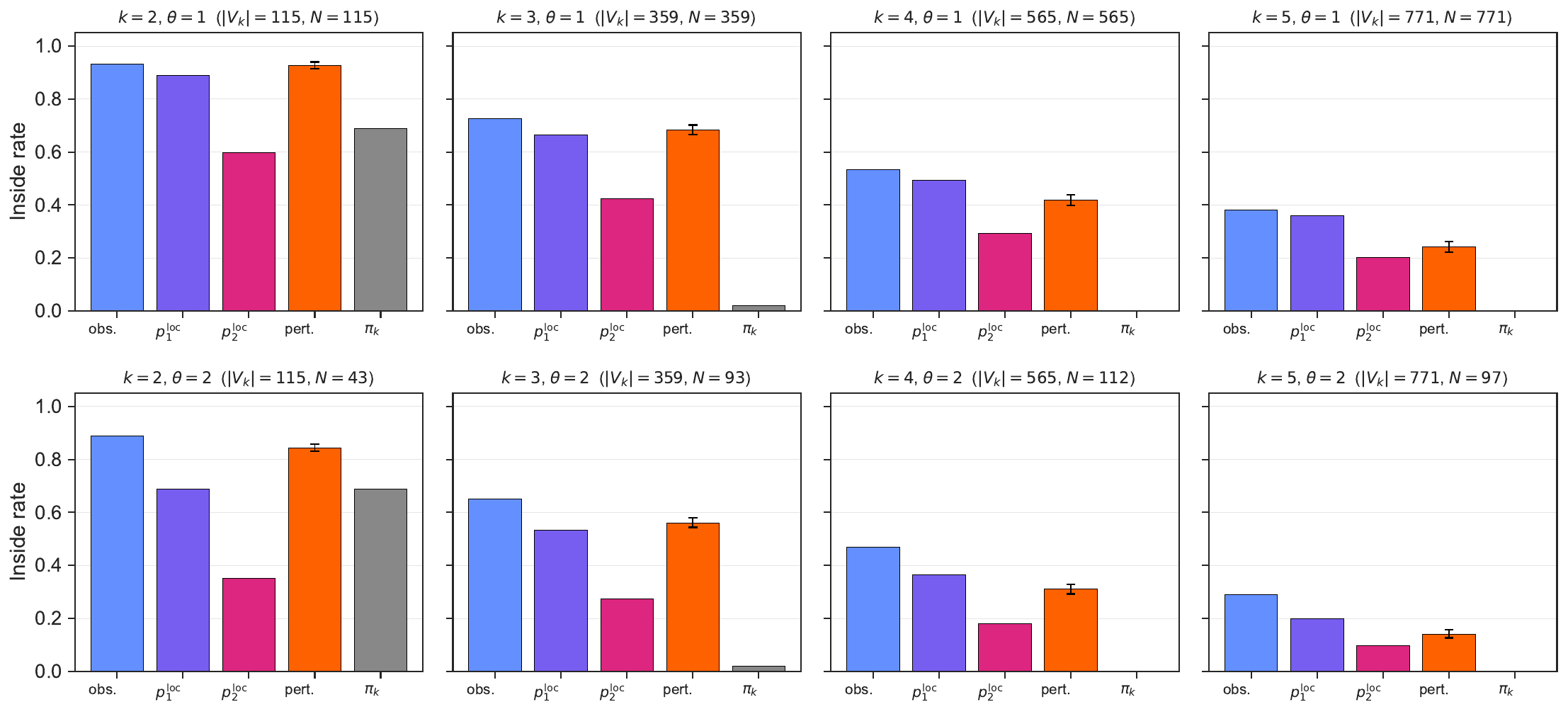}
\caption{\textbf{Inside-rate statistics across all $(k, \theta)$.} Each cell shows observed inside rate, first- and second-neighbour random-walk expectations, perturbation null mean (with $95\%$ band), and analytic coverage baseline $\pi_k$. Top row: untrimmed occupancy. Bottom row: trimmed occupancy. Trimming reduces the size of the normative set $|\mathcal{N}_k|$ (reported in panel titles) but generally amplifies the gaps with the random-walk and perturbation nulls.}
\label{fig:appendix-A1}
\end{figure}

\begin{figure}[H]
\centering
\includegraphics[width=0.98\linewidth]{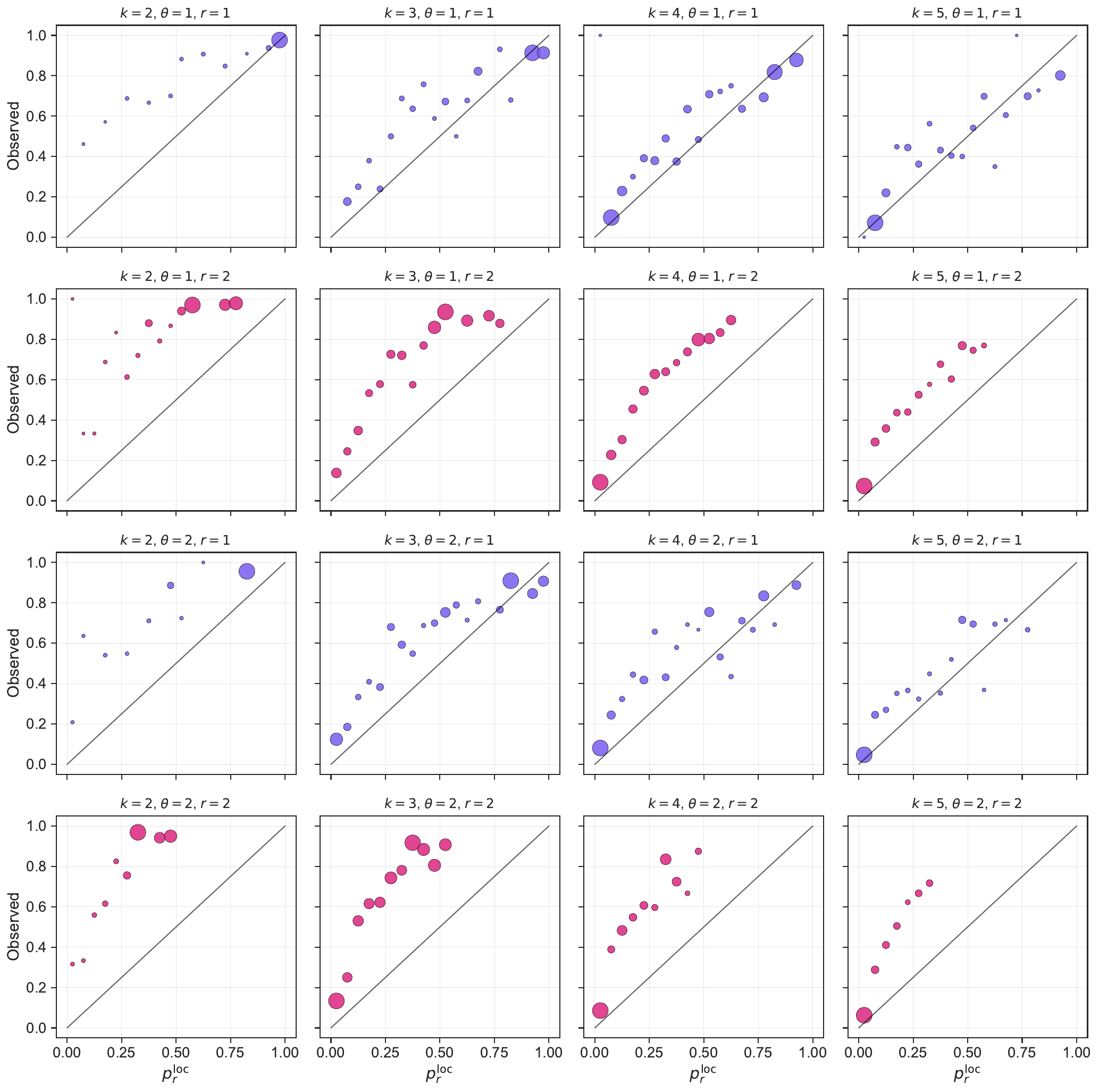}
\caption{\textbf{Per-respondent observed inside rate vs.\ local random-walk expectation, all $(k,\theta,r)$.} Columns: $k\in\{2,3,4,5\}$. Rows from top: $(\theta=1, r=1)$, $(\theta=1, r=2)$, $(\theta=2, r=1)$, $(\theta=2, r=2)$. Marker size is proportional to the number of respondents in each $x$-bin.}
\label{fig:appendix-A2}
\end{figure}

\begin{figure}[H]
\centering
\includegraphics[width=0.7\linewidth]{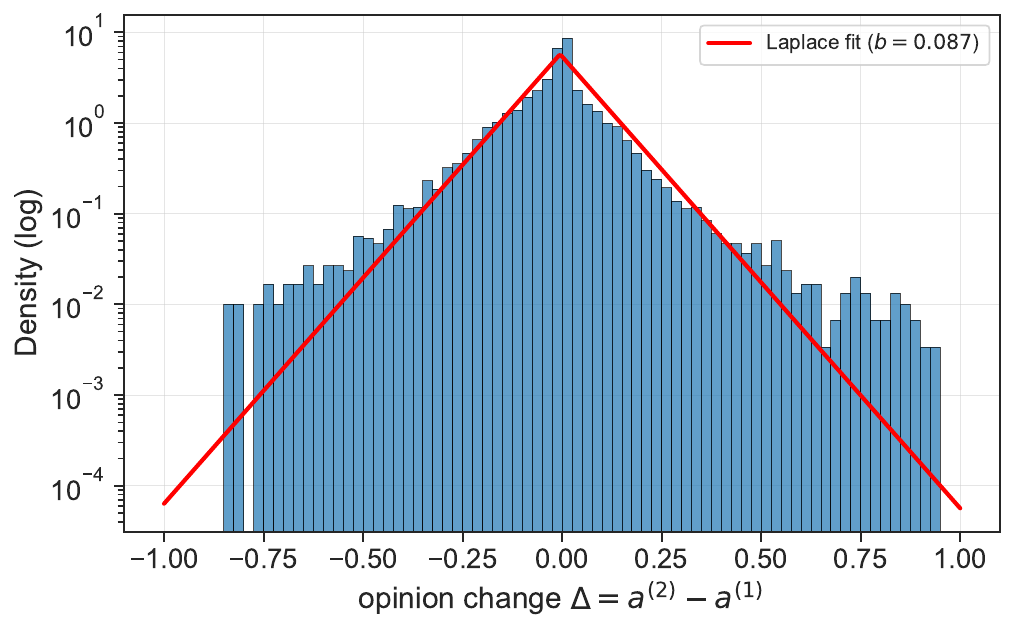}
\caption{\textbf{Pooled distribution of issue-level opinion change.} Histogram on log scale across all respondents and issues, with a Laplace MLE overlay. The linear shape on the log scale is the Laplace signature; the empirical fit gives location $\mu \approx 0$ and scale $b \approx 0.087$.}
\label{fig:laplace}
\end{figure}

\begin{figure}[H]
\centering
\includegraphics[width=0.98\linewidth]{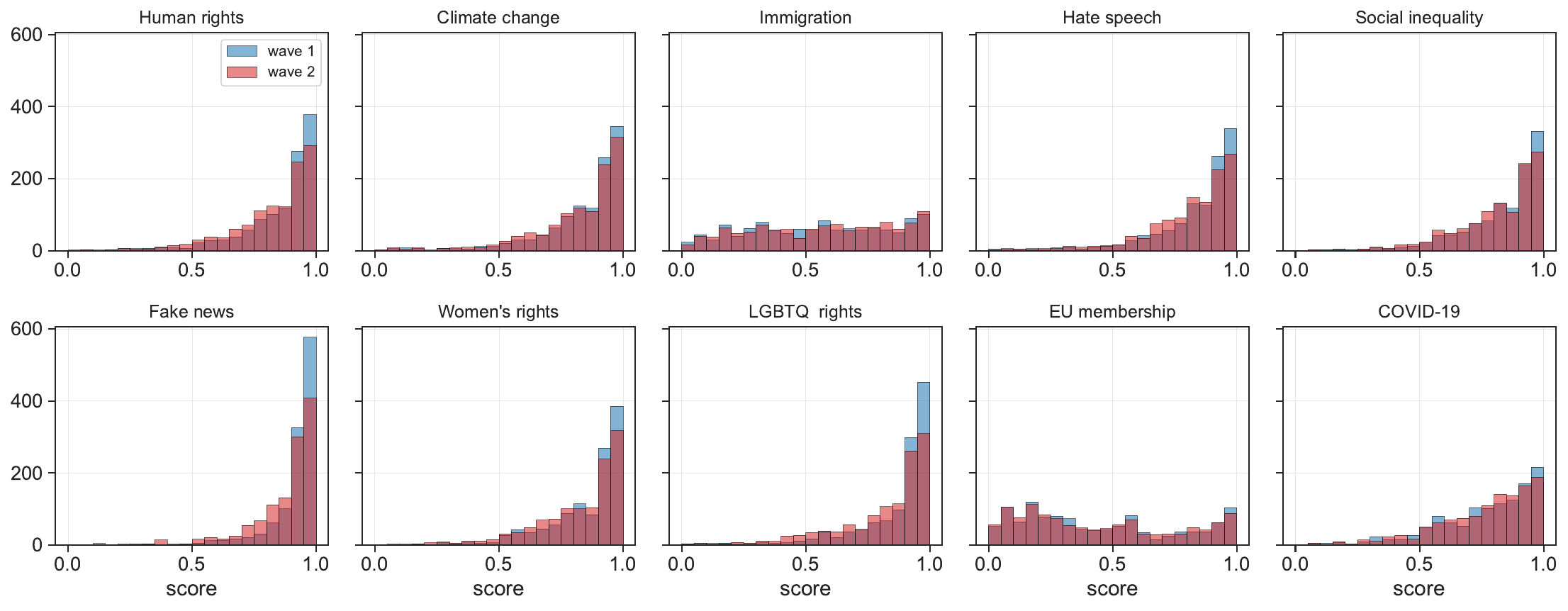}
\caption{\textbf{Per-issue attitude distributions at wave 1 and wave 2.} The distributions are markedly non-uniform: most respondents support the progressive framing of each statement, with the exception of EU membership and Immigration, which have inverted distributions due to opposing-direction framing in the original survey.}
\label{fig:marginals}
\end{figure}

\begin{figure}[H]
\centering
\includegraphics[width=0.98\linewidth]{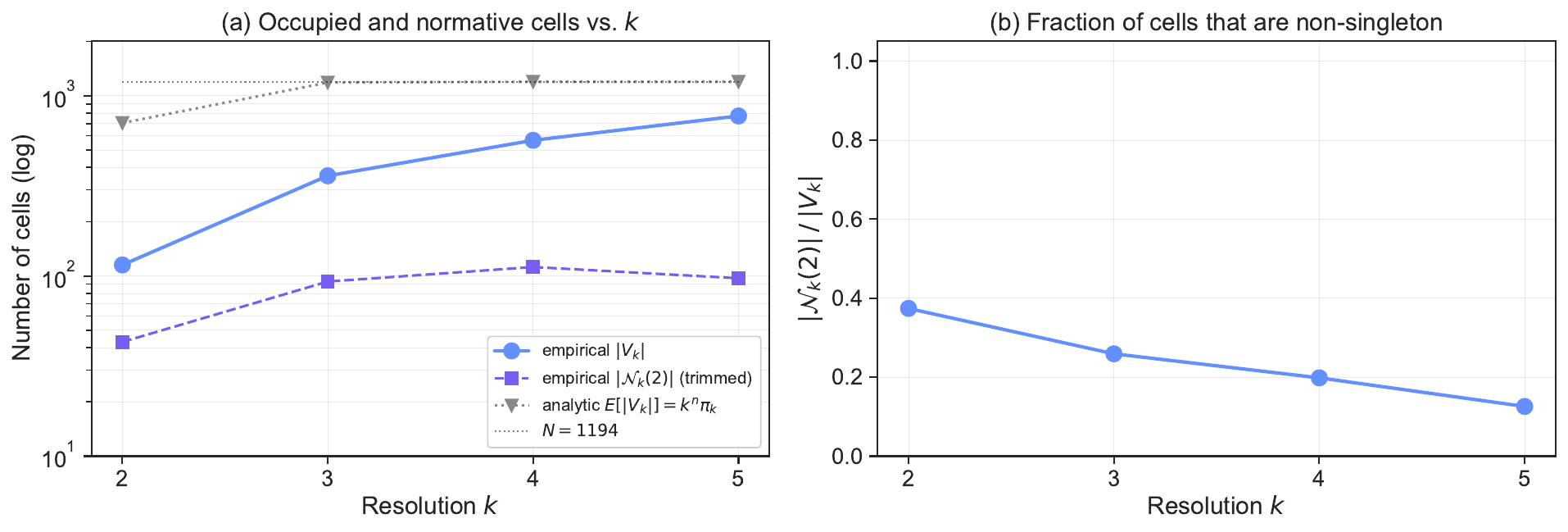}
\caption{\textbf{Lattice statistics across resolutions.} (a) Number of empirically occupied cells $|C_k|$ and trimmed normative cells $|\mathcal{N}_k(2)|$ alongside the analytic prediction $E[|C_k|] = k^n \pi_k$ under uniform random placement. The empirical points sit below the analytic curve at every resolution, reflecting baseline concentration. (b) Fraction of occupied cells that are non-singleton, $|\mathcal{N}_k(2)|/|C_k|$, as a function of $k$.}
\label{fig:lattice-stats}
\end{figure}

\begin{figure}[H]
\centering
\includegraphics[width=0.98\linewidth]{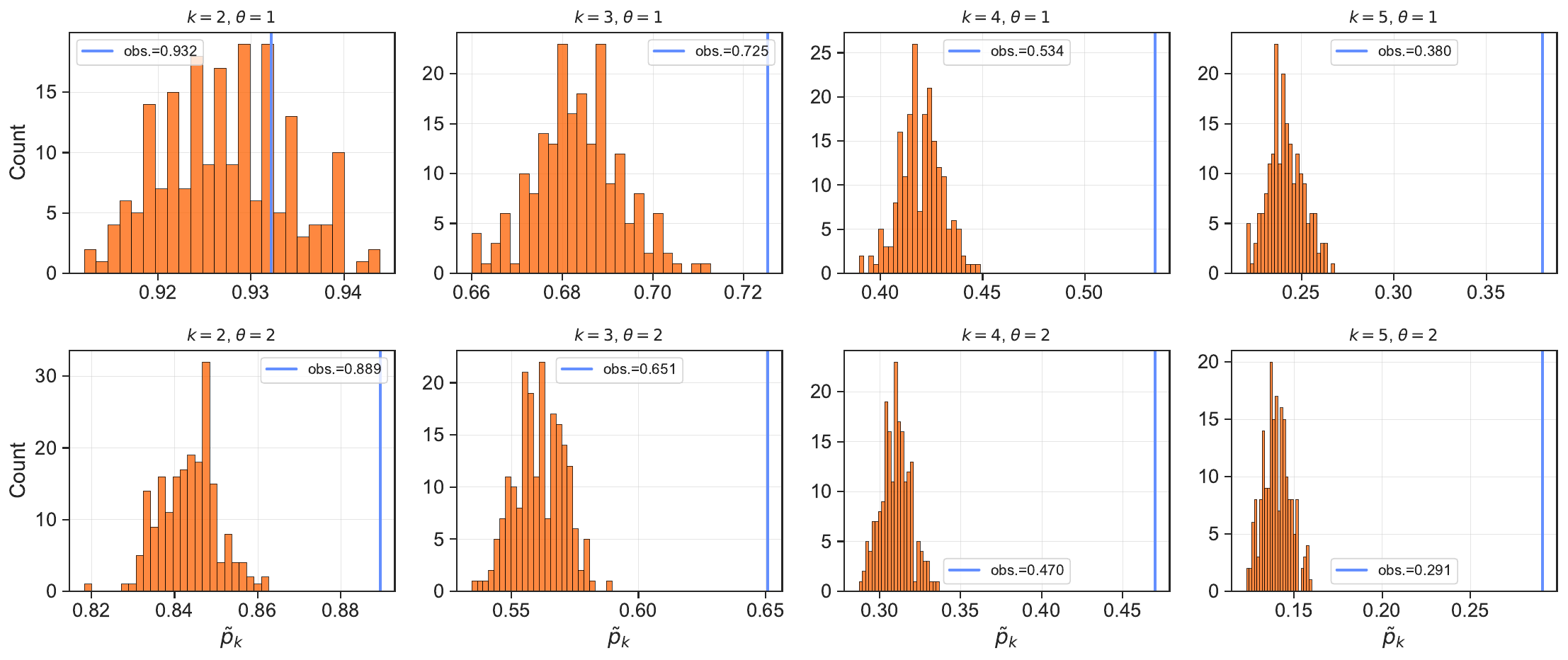}
\caption{\textbf{Distributions of perturbation-null inside rates ($200$ permutations).} Each panel shows the histogram of $\tilde{p}_k(\theta)$ under independent issue-wise permutation of opinion changes. The vertical line marks the empirical $\hat{p}_k(\theta)$. Top row: untrimmed. Bottom row: trimmed.}
\label{fig:perturbation-dists}
\end{figure}

\bibliographystyle{plain}
\bibliography{sample}

@article{Lorenz2011,
author = {Lorenz, J. and Rauhut, H. and Schweitzer, F. and Helbing, D.},
doi = {10.1073/pnas.1008636108},
issn = {0027-8424},
journal = {Proceedings of the National Academy of Sciences},
number = {22},
pages = {9020--9025},
title = {{How social influence can undermine the wisdom of crowd effect}},
volume = {108},
year = {2011}
}

@article{Bond1996,
author = {Bond, Rod and Smith, Peter B.},
doi = {10.5040/9781472551351.ch-016},
journal = {Psychological Bulletin},
number = {I},
pages = {111--137},
title = {{Culture and Conformity: A Meta-Analysis of Studies Using Asch's (1952b, 1956) Line Judgement Task}},
volume = {119},
year = {1996}
}

@article{Watts2007,
author = {Watts, Duncan J and Dodds, Peter Sheridan and Dodds, Peter Sheridan},
journal = {Journal of Consumer Research},
number = {4},
pages = {441--458},
title = {{Influentials, networks, and public opinion formation}},
volume = {34},
year = {2007}
}

@article{Solomon1955,
author = {Asch, Solomon},
doi = {10.1038/scientificamerican1155-31},
issn = {0036-8733},
journal = {Scientific American},
number = {5},
pages = {31--35},
title = {{Opinions and social pressure}},
url = {http://www.freepatentsonline.com/WO2003027106.html},
volume = {193},
year = {1955}
}

@article{BikhchandaniSushilHirshleifer1992,
author = {{Bikhchandani, Sushil, Hirshleifer}, David and and Welch, Ivo},
doi = {10.1127/njgpa/2016/0578.Beardmore},
isbn = {0542792281},
journal = {Journal of Political Economy},
number = {5},
pages = {992--1026},
title = {{A theory of fads, fashion, custom, and cultural change as informational cascades}},
volume = {100},
year = {1992}
}

@article{nadeau1993new,
  title={New evidence about the existence of a bandwagon effect in the opinion formation process},
  author={Nadeau, Richard and Cloutier, Edouard and Guay, J-H},
  journal={International Political Science Review},
  volume={14},
  number={2},
  pages={203--213},
  year={1993},
  publisher={Sage Publications Sage CA: Thousand Oaks, CA}
}

@article{marsh1985back,
  title={Back on the bandwagon: The effect of opinion polls on public opinion},
  author={Marsh, Catherine},
  journal={British Journal of Political Science},
  volume={15},
  number={1},
  pages={51--74},
  year={1985},
  publisher={Cambridge University Press}
}

@incollection{margetts2015political,
  title={Political turbulence},
  author={Margetts, Helen and John, Peter and Hale, Scott and Yasseri, Taha},
  booktitle={Political Turbulence},
  year={2015},
  publisher={Princeton University Press}
}

@article{mutz1992mass,
  title={Mass media and the depoliticization of personal experience},
  author={Mutz, Diana C},
  journal={American Journal of Political Science},
  pages={483--508},
  year={1992},
  publisher={JSTOR}
}

@article{Fleitas1971,
  title={Bandwagon and Underdog Effects in Minimal-Information Elections},
  author={Fleitas, Daniel W.},
  journal={American Political Science Review},
  volume={65},
  number={2},
  pages={434-438},
  year={1971},
  publisher={}
}

@article{Chatterjee2021,
  title={Voting for the underdog or jumping on the bandwagon? Evidence from India’s exit poll ban},
  author={Chatterjee, Somdeep and Kamal, Jai},
  journal={Public Choice},
  volume={188},
  number={3-4},
  pages={431-453},
  year={2021},
  publisher={}
}

@article{Kenney1996,
  title={The Psychology of Political Momentum},
  author={Kenney, Patrick J. and Rice, Tom W.},
  journal={Political Research Quarterly},
  volume={47},
  number={},
  year={1996},
  publisher={}
}

@article{Mehrabian1998,
  title={Effects of poll reports on voter preferences},
  author={Mehrabian, Albert},
  journal={Journal of Applied Social Psychology},
  volume={28},
  number={},
  pages={2119–2130},
  year={1999},
  publisher={}
}

@article{pantazipublic,
  title={How public opinion shapes individual policy preferences: do the conformists constrain the contrarians?},
  author={Pantazi, Myrto and Camargo, Chico and Margetts, Helen and John, Peter and Hale, Scott},
  publisher={OSF}
}

\end{document}